\documentclass[11pt]{article}
\usepackage[margin=1in]{geometry}
\usepackage{enumerate}
\usepackage[linesnumbered,vlined,noend,boxed,algoruled]{algorithm2e}
\usepackage{amsmath,amsthm,amssymb}
\usepackage{graphicx}
\usepackage{paralist}
\usepackage{verbatim}
\usepackage{xcolor}
\usepackage{multicol}
\usepackage{pdfpages}
\newcommand{\bv}[1]{\mathbf{#1}}

\usepackage[colorlinks=true,
            linkcolor=red,
            urlcolor=blue,
            citecolor=gray]{hyperref}

\newtheorem{theorem}{Theorem}[section]

\newtheorem{definition}{Definition}

\newtheorem{lemma}[theorem]{Lemma}

\newcommand{\E}{\mathbb{E}}

\newcommand{\BO}{\mathcal{O}}

\begin{document}

\title{Distributed House-Hunting in Ant Colonies}
\author{
Mohsen Ghaffari
\and
Cameron Musco
\and
Tsvetomira Radeva
\and
Nancy Lynch
\and
 \\ \{ghaffari, cnmusco, radeva, lynch\}@csail.mit.edu, MIT
}
\date{}
\maketitle

\begin{abstract}
We introduce the study of the \emph{ant colony house-hunting} problem from a distributed computing perspective. When an ant colony's nest becomes unsuitable due to size constraints or damage, the colony must relocate to a new nest. The task of  identifying and evaluating the quality of potential new nests is distributed among all ants. The ants must additionally reach consensus on a final nest choice and the full colony must be transported to this single new nest. Our goal is to use tools and techniques from distributed computing theory in order to gain insight into the house-hunting process.

We develop a formal model for the house-hunting problem inspired by the behavior of the \emph{Temnothorax} genus of ants. We then show a $\Omega(\log n)$ lower bound on the time for all $n$ ants to agree on one of $k$ candidate nests. We also present two algorithms that solve the house-hunting problem in our model. The first algorithm solves the problem in optimal $\BO(\log n)$ time but exhibits some features not characteristic of natural ant behavior. The second algorithm runs in $\BO(k \log n)$ time and uses an extremely simple and natural rule for each ant to decide on the new nest.
\end{abstract}

\section{Introduction}
\label{sec:intro}

Some recent work in distributed computing theory has focused on biological problems inspired by algorithmic tasks carried out by ant colonies; for example,  \cite{emek2014solving, feinerman2012memory, feinerman2012collaborative, lenzen2014trade} study collaborative food foraging and \cite{cornejo2014task} models and proposes algorithms for task allocation within a colony.  
Often computer scientists study biologically-inspired algorithms with the aim of engineering better computing systems.
An alternative line of work uses tools developed to analyze and design distributed computer systems to better understand the behavior of biological systems \cite{feinerman2013theoretical}.
In this paper, we follow the second approach in an attempt to better understand the house-hunting behavior of the \emph{Temnothorax} ant. The essence of this process is a collaborative search, consensus decision, and relocation of the entire ant colony to a new home, with an emphasis on the decision making component.

We first model the ant colony and its computational constraints, and formally define an algorithmic problem that it solves. 
The challenge is to develop a model that is simultaneously tractable to theoretical analysis and close enough to reality to give meaningful insights into ant behavior.
For example: Is ant behavior in some sense optimal given their biological constraints?  
One hypothesis may be that behavior has been optimized through evolution. Lower bounds matching the performance of algorithms seen in nature can provide mathematical support for this hypothesis. 
Alternatively, results may show that behavior is far from optimal, suggesting the existence of hidden constraints or goals yet to be fully understood. Can we identify why certain behavioral patterns have developed and what environmental and biological constraints they are adaptations to?


One motivation for studying the house-hunting process of \emph{Temnothorax} ants is that it has received significant attention from biologists. Various aspects of the process have been studied, including ants' preferences and capabilities \cite{healey2008effect, sasaki2013ants}, the general structure of the algorithms used by the ants \cite{mallon2001individual, pratt2002quorum}, variations due to colony size \cite{franks2006decision}, and trade-offs between the speed and accuracy of the process \cite{pratt2006tunable}.
This work forms a wide basis of scientifically tested hypotheses that we can use for the foundation of our model and inspiration for our algorithms.

From a distributed computing perspective, house-hunting is closely related to the fundamental and well-studied problem of
consensus \cite{fischer1985impossibility, lamport1998part}.
This makes the problem conceptually different from other ant colony inspired problems studied by computer scientists. 
Task allocation and foraging are both intrinsically related to parallel optimization. The main goal is to divide work optimally amongst a large number of ants 
in a process similar to load balancing. This is commonly achieved using random allocation or \emph{negative feedback} \cite{astrom2010feedback} against work that has already been completed.
In contrast, the house-hunting problem is a decision problem in which all ants must converge to the same choice. Both in nature and in our proposed algorithms, this is achieved through \emph{positive feedback} \cite{astrom2010feedback}, by reinforcing viable nest candidates until a single choice remains. At a high level, our work is related to previous work on rumor spreading in biological populations \cite{feinerman2013efficient}.

\subsection{The House-Hunting Process}
\label{sec:hh}
\emph{Temnothorax} ants live in fragile rock crevices that are frequently destroyed. It is crucial for colony survival to quickly find and move to a new nest after their home is compromised. This process is highly distributed and involves several stages of searching for nests, assessing nest candidates, recruiting other ants to do the same, and finally, transporting the full colony to the new home.

In the search phase, some ants begin searching their surroundings for possible new nests. Experimentally, this phase has not been studied much; it has been assumed that ants encounter candidate nests fairly quickly through random walking. In the assessment phase, each ant that arrives at a new nest evaluates it based on various criteria, e.g., whether the nest interior is dark and therefore likely free of holes, and whether the entrance to the nest is small enough to be easily defended. These criteria may have different priorities \cite{healey2008effect, sasaki2013ants} and, in general, it is assumed that nest assessments by an individual ant are not always precise or rational \cite{sasaki2011emergence}. After some time spent assessing different nests, going back to the old nest and searching for new nests, an ant becomes sufficiently satisfied with some nest and moves on to the recruitment phase, which consists of \emph{tandem runs} -- one ant leading another ant from the old to a new nest. The recruited ant learns the candidate nest location and can assess the nest itself and begin performing tandem runs if the nest is acceptable. 

At this point many nest sites may have ants recruiting to them, so a decision has to be made in favor of one nest. The ants must solve the classic distributed computing problem of consensus. One strategy that ants are believed to use is a \emph{quorum threshold} \cite{pratt2005quorum, pratt2002quorum} -- a  threshold of the number of ants in a candidate nest, that, when exceeded, indicates that the nest should be chosen as the new home. Each time an ant returns to the new nest, it evaluates (not necessarily accurately) whether a quorum has been reached. If so, it begins the transport phase -- picking up and carrying other ants from the old to the new nest. These transports are generally faster than tandem runs and they conclude the house-hunting process by bringing the rest of the colony to the new nest.

\subsection{Main Results and Organization}
\label{sec:results}
Our main results are a mathematical model of the house-hunting process, a lower bound on the number of rounds required by any algorithm solving the house-hunting problem in the given model, and two house-hunting algorithms.

Our model (Section \ref{sec:model}) is based on a synchronous model of execution with $n$ probabilistic ants and communication limited to one ant leading another ant (tandem run or transport), chosen randomly from the ants at the home nest, to a candidate nest. Ants can also search for new nests by choosing randomly among all $k$ candidate nests. We do not model the time for an ant to find a nest or to lead a tandem run; each of these actions are assumed to take one round.

Our lower bound (Section \ref{sec:lower}) states that, under this model, no algorithm can solve the house-hunting problem in time sub-logarithmic in the number of ants. The main proof idea is that, in any step of an algorithm's execution, with constant probability, an ant that does not know of the location of the eventually-chosen nest remains uninformed. Therefore, with high probability, $\Omega(\log n)$ rounds are required to inform all $n$ ants. This technique closely resembles lower bounds for rumor spreading in a complete graph, where the rumor is the location of the chosen nest \cite{karp2000randomized}.

Our first algorithm (Section \ref{sec:optimal}) solves the house-hunting problem in asymptotically optimal time. The main idea is a typical example of positive feedback: each ant leads tandem runs to some suitable nest as long as the population of ants at that nest keeps increasing; once the ants at a candidate nest notice a decrease in the population, they give up and wait to be recruited to another nest. With high probability, within $\BO(\log n)$ rounds, this process converges to all $n$ ants committing to a single winning nest. Unfortunately, this algorithm relies heavily on a synchronous execution and on the ability to precisely count nest populations, suggesting that the algorithm is susceptible to perturbations of our model and most likely does not match real ant behavior. 

The goal of our second algorithm (Section \ref{sec:natural}) is to be more natural and resilient to perturbations of the environmental parameters and ant capabilities. The algorithm uses a simple positive-feedback mechanism: in each round, an ant that has located a candidate nest recruits other ants to the nest with probability proportional to its current population. We show that, with high probability, this process converges to all $n$ ants being committed to one of the $k$ candidate nests within $\BO(k \log n)$ rounds.  While this algorithm is not optimal, except when $k$ is assumed to be constant, it exhibits a much more natural process of converging to a single nest. In Section \ref{sec:discussion}, we discuss in more detail possible modifications to the algorithm and various perturbations and faults to which it is resilient. Such robustness criteria are necessary in nature and generally desirable for distributed algorithms.

\section{Model}
\label{sec:model}
Here we present a simple model of \emph{Temnothorax} ants behavior that is tractable to rigorous analysis, yet rich enough to provide a starting point for understanding real ant behavior.

The environment consists of a home nest, denoted $n_0$, along with $k$ candidate new nests, identified as $n_i$ for $i \in \{1,\cdots,k\}$. Each nest $n_i$ is assigned a quality $q(i) \in Q$, for some set $Q$. Throughout this paper we let  $Q= \{0,1\}$, with quality $0$ indicating an unsuitable nest, and $1$ a suitable one. Additionally, we assume that there is always at least one nest with $q(i) = 1$.

The colony consists of $n$ identical probabilistic finite state machines, representing the ants. We assume $n$ is significantly larger than $k$, with $k = O(n/\log n)$. Additionally, ants know the value of $n$ but not $k$, so the state machines may be parameterized by $n$ but must be uniform for all $k$. This assumption is based on evidence that real \emph{Temnothorax} ants and other species are able to estimate colony size and change their behavior in response \cite{dornhaus2006colony,cronin2013synergy}.

The general behavior of the state machines is unrestricted but their interactions with the environment and with other ants are limited to the high-level functions $\bv{search()}$, $\bv{go()}$, and $\bv{recruit()}$, defined below. Ants execute synchronously in numbered rounds starting with round $1$. In each round, each ant may perform unlimited local computation (transition through an arbitrary sequence of states), along with exactly one call to one of the functions: $\bv{search()}$, $\bv{go()}$, or $\bv{recruit()}$.

At the end of each round $r$, each ant $a$ is located at a nest, denoted by
$\ell(a,r) \in \{0,1,\cdots,k\}$; initially, before round $1$, all ants are located at the home nest. The value of $\ell(a,r)$ is set by the calls to $\bv{search()}$, $\bv{go()}$, or $\bv{recruit()}$ made by the ant in round $r$. Let $c(i,r) = |\{a|\ell(a,r) = i\}|$ denote the number of ants located in nest $n_i$ at the end of round $r$. 

In each round, each ant $a$ performs exactly one call to the following functions:
\begin{itemize}

\item $\bv{search()}$: Returns a triple $<i, q(i), c(i,r)>$ where $i$ is chosen uniformly at random from $\{1, \cdots, k\}$. Sets  $\ell(a,r) := i$. This function represents ant $a$ searching randomly for a nest; the return value of the function includes the nest index, the nest's quality and the number of ants at the nest.
	
\item $\bv{go(i)}$: Takes input $i \in \{1,\cdots,k\}$ such that there exists a round $r'<r$ in which $\ell(a,r') = i$. Returns $c(i,r)$. Sets $\ell(a,r) := i$. The function represents ant $a$ revisiting a candidate nest $n_i$; the function returns the number of ants at nest $n_i$ at the end of round $r$.
	
\item $\bv{recruit(b,i)}$: Takes input $b \in \{0,1\}$ and $i \in \{1,\cdots,k\}$ such that there exists a round $r'<r$ in which $\ell(a,r') = i$. Returns a pair $<j, c(0,r)>$ where $j \in \{1, \cdots, k\}$. Sets $\ell(a,r) := 0$. The return value $j$ is determined as follows.
Let $R$ be the set of all ants that call $recruit(\cdot, \cdot)$, and let $P$ be a uniform random permutation of all ants in $R$. Let $S \subseteq R$ be the set of ants that call $recruit(1,\cdot)$.

\end{itemize}

\begin{algorithm*}[h!]
\caption{Return value $j$ of $\bv{recruit(\cdot,\cdot)}$ for all ants $a \in R$}	
\label{model:recruit}
\begin{multicols}{2}
$M$: a set of pairs $(a,a')$ of ants, initially $\emptyset$ \\
\For{$i = 1$ to $|P|$}{
	\If{$a_{P(i)} \in S$ and $(\cdot, a_{P(i)}) \not\in M$}{
		$a' := $ uniform random ant from $R$ \\
		\If{$(a',\cdot) \not\in M$ and $(\cdot, a') \not\in M$}{
			$M := M \cup (a_{P(i)},a')$
		}
	}
}
\For{$i = 1$ to $|P|$}{
	\If{$(a^*, a_{P(i)}) \in M$}{
		\Return{$j$ to ant $a_{P(i)}$ where $j$ is \\input to $recruit(1,j)$ called by $a^*$} 
	}
	\textbf{else} 
		 \Return{$j$ to ant $a_{P(i)}$ where $j$ is \\input to $recruit(\cdot,j)$ called by $a_{P(i)}$}
}
\end{multicols}
\vspace{1.5mm}
\end{algorithm*}

In short, all actively recruiting ants in $S$ randomly choose an ant to recruit. $P$ simply serves as tie breaker to avoid conflicts between recruitments. It is important to note that this process is not a distributed algorithm executed by the ants, but just a modeling tool to formalize the idea of ants recruiting other ants randomly without introducing dependency chains between the ordered pairs of recruiting and recruited ants. Algorithm \ref{model:recruit} can be thought of as a centralized process run by the environment of the home nest in order to pair ants appropriately. We believe our results also hold under other natural models for randomly pairing ants. 

An ant \emph{recruits successfully} if it is the recruiting ant (first element) in one of the pairs in $M$.
If an ant recruits successfully or is not recruited, $\bv{recruit()}$ just returns the input nest id $i$. Otherwise, it returns the id of the nest that the ant is recruited to.

We intentionally define $\bv{recruit()}$ to locate ants back to the home nest, where recruitment happens. Additionally, $\bv{go(i)}$ is only applicable to a candidate nest $n_i$ for $i \neq 0$, so calling $\bv{recruit()}$ is the only way to return to the home nest. Since ants are required to call one of the three functions in each round, if an ant does not search (call $\bv{search()}$) or go to a nest (call $\bv{go()}$), it is required to stay at the home nest. 
In this way, all ants located in the home nest participate in recruitment in each round, either actively recruiting ($b=1$) or waiting to be recruited ($b=0$).

\emph{Recruitment} encompasses both the tandem runs and direct transport behavior observed in \emph{Temnothorax} ants. Since direct transport is only about three times faster than tandem walking \cite{pratt2010nest}, and since we focus on asymptotic behavior, we do not model this action separately.	

Next, we prove a general statement about the recruitment process that will be used in the proofs of our algorithms.

\begin{lemma}
\label{lem:recruit-succeed}
	Let $a$ be an arbitrary ant that executes $recruit(1,\cdot)$ in some round $r$, and assume $c(0,r) \geq 2$. Then, with probability at least $1/16$, $(a,\cdot) \in M$; that is, ant $a$ succeeds in recruiting another ant.
\end{lemma}
	
\begin{proof}
	One requirement for ant $a$ to succeed recruiting is that it does not get recruited by another ant positioned earlier in the random permutation $P$ than $a$. With probability $1/2$ ant $a$ is located in the first half of the random permutation, so conditioning on that, the probability that $a$ is not recruited by some ant with higher precedence in some round $r$ is at least:
	\begin{equation*}
		\left(1-\frac{1}{c(0,r)}\right)^{c(0,r)/2} \geq \frac{1}{4^{(1/2)}} \geq \frac{1}{2},
	\end{equation*}
where $c(0,r) \ge 2$ is the number of ants at the home nest in round $r$.
	
	Another requirement for ant $a$ to succeed recruiting is to choose to recruit an ant that has not itself successfully recruited and has not been recruited by the preceding ants. Recall that we conditioned on ant $a$ being located in the first half of the permutation $P$. With probability $1/2$, ant $a$ chooses to recruit an ant $a'$ located in the second half of the permutation; therefore, $a'$ appears after $a$ in the permutation and so $a'$ has not recruited yet when $a$ recruits. With probability at least $\left(1-\frac{1}{c(0,r)}\right)^{c(0,r)/2} \ge 1/2$, ant $a'$ is not recruited by any ant other than $a$ in round $r$.
	 
	 Therefore, in total, the probability that ant $a$ recruits successfully is at least: $1/2$ to be located in the first half of the permutation, $1/2$ to not be recruited by another ant, $1/2$ to recruit an ant in the second half of the permutation, and $1/2$ for this chosen ant to not be recruited by any other ant. This gives a total probability of at least $1/16$. 
\end{proof}

\paragraph{Problem Statement:}
An algorithm $\mathcal{A}$ solves the \textsc{HouseHunting} problem with $k$ nests in $T\in\mathbb{N}$ rounds with probability $1-\delta$, for $0 < \delta \leq 1$, if with probability $1-\delta$, taken over all executions of $\mathcal{A}$, there exists a nest $i \in \{1,\cdots,k\}$ such that $q(i)=1$ and $\ell(a,r) = i$ for all ants $a$ and for all rounds $r \geq T$.

\section{Lower Bound}
\label{sec:lower}

In this section, we present a lower bound on the number of rounds required for an algorithm to solve the house-hunting problem. The key idea of the proof is similar to the lower bounds on spreading a rumor in a complete graph \cite{karp2000randomized} where neighbors contact each other randomly. Assuming a  house-hunting process with a single good nest, its location represents the rumor to be spread among all ants and communication between random neighbors is analogous to the recruiting process.

Assume only a single nest, $n_w$, has quality $1$, and so it is the only option for the ants to relocate to. Additionally, assume that each ant is able to recognize nest $n_w$ once it knows its id. These assumptions only restrict the environment and add to the ants' power, so a lower bound under these assumptions is sufficient. Note that each ant has only two ways of reaching a nest: (1) by searching, or (2) by getting recruited to it.
Throughout this section, an ant is considered to be \emph{informed} if it know the id of the winning nest $n_w$; otherwise, an ant is considered to be \emph{ignorant}.

\begin{lemma}
\label{lem:ignorant}
	With probability at least $1/4$ an ant that is ignorant at the beginning of round $r$ remains ignorant at the end of round $r$.
\end{lemma}
\begin{proof}
	In round $r$, ant $a$ may either (1) not participate in recruitment or search (2) participate in recruitment at the home nest or (3) search. In the first case, ant $a$ does not call either $\bv{recruit()}$ or $\bv{search()}$ and is guaranteed to remain ignorant. 
	
	In the second case, ant $a$ calls $\bv{recruit(0,\cdot)}$. Let $X_r$ be the number of informed ants at the beginning of round $r$ that are recruiting at the home nest to the winning nest; these ants are calling $\bv{recruit(1,w)}$. Then, the probability that ant $a$ remains ignorant after round $r$ is at least:
	\begin{equation*}
		\left(1-\frac{1}{c(0,r)}\right)^{X_r} \geq \left(1-\frac{1}{c(0,r)}\right)^{c(0,r)} \geq \frac{1}{4}
	\end{equation*}
	where $c(0,r) \ge 2$ is the number of ants at the home nest at the beginning of round $r$. Note that if $c(0,r) < 2$, ant $a$ is forced to recruit itself, so it remains ignorant.
	
In the third case, for $k \ge 2$, the probability that ant $a$ remains ignorant after one round of searching is $(k-1)/k \geq 1/2$. 

Therefore, the overall probability that an ignorant ant remains ignorant in round $r$, after either searching or participating in recruitment, is at least $1/4$ (since $1/2 > 1/4$).
\end{proof}

Let random variable $\bar{X_r}$ denote the number of ignorant ants at the beginning of round $r$. In order for the algorithm to solve the house hunting problem, it is necessary that $\bar{X_r} = 0$ for some round $r$. Otherwise, there is at least one ant that is not committed to the winning nest $n_w$.

\begin{theorem}
\label{thm:lower}
	For any constant $c >0$, let $\mathcal{A}$ be an algorithm that solves the \textsc{HouseHunting} problem with $k \ge 2$ nests in $T$ rounds with probability at least $1/n^c$. Then, $T = \Omega(\log n)$.
\end{theorem}
\begin{proof}
	
By Lemma \ref{lem:ignorant}, with probability at least $1/4$ an ant that is ignorant at the beginning of round $r$ remains ignorant at the end of round $r$.
The probability that an ignorant ant remains ignorant for $r = (\log_4 n)/2 - \log_4(12c)$ consecutive rounds is at least $(1/4)^r \ge 12c/\sqrt{n}$, where $c>1$ is an arbitrary constant\footnote{We assume a constant is asymptotically smaller than $\sqrt{n}$.}.
So the expected number of ignorant ants at the beginning of round $r = (\log_4 n)/2 - \log_4(12c)$ is $\E[\bar{X}_r] \geq (12c/\sqrt{n})(n)= 12c\sqrt{n}$.

Note that the random variable $\bar{X}_r$ can be expressed as a sum of indicator random variables $\bar{X}_r^a$, where for each ant $a$, $\bar{X}_r^a$ is $1$ if ant $a$ is ignorant, and $0$ if it is informed, at the beginning of round $r$. Unfortunately, these indicator random variables are not independent, so we cannot directly apply a Chernoff bound. However, we can still show that their sum is bounded by defining random variables that dominate the $\bar{X}_r^a$ variables. 

Let random variable $Y_r^a$ be a random variable such that $P[Y_r^a = 1] = (1/4)^r$ for each $a$ and each $r$. Let $Y_r$ be the sum of these independent random variables for all ants $a$. Note that $\E[Y_r] = n(1/4)^r$, so for $r = (\log_4 n)/2 - \log_4(12c)$, $\E[Y_r] = 12c\sqrt{n}$. Therefore, by a Chernoff bound:

\begin{equation*}
	P[Y_r < 6c\sqrt{n}] < e^{-(12c\sqrt{n})/12} \leq e^{-c\ln n} \leq \frac{1}{n^c}.
\end{equation*}

By Lemma \ref{lem:ignorant}, the probability that an uninformed ant remains uninformed is at least $1/4$, regardless of the other ants' actions in round $r$. Therefore, $P[\bar{X}_r^a = 1 | \text{all } \bar{X}_r^{a'} \text{ for $a'$ before $a$ in $P$}] \geq (1/4)^r$, where $P$ is the random permutation used to model the recruitment process. So, by the definition of $Y_r^a$:

\begin{equation*}
	P[\bar{X}_r^a = 1 | \text{all } \bar{X}_r^{a'} \text{ for $a'$ before $a$ in $P$}] > \left(\frac{1}{4}\right)^r = P[Y_r^a = 1]
\end{equation*}

By Lemma 1.18 in \cite{doerr2011theory}[Chapter 1] (also stated as Lemma \ref{lem:dependent-chernoff} in the Appendix), $P[\bar{X}_r < x] \leq P[Y_r < x]$ for any $x \leq n$. In particular, for $r = (\log_4 n)/2 - \log_4(12c)$, it follows that :

\begin{equation*}
	P[\bar{X}_r < 6c\sqrt{n}] \leq P[Y_r < 6c\sqrt{n}] < \frac{1}{n^c}
\end{equation*}

Therefore, with probability at least $1-1/n^c$, at least $6c\sqrt{n}$ ants are ignorant at the end of round $r = (\log_4 n)/2 - \log_4(12c)$. These ants are not informed of the id of the  winning nest, and so cannot be located at this nest. 

Since algorithm $\mathcal{A}$ solves the \textsc{HouseHunting} problem with probability at least $1/n^c$ in $T$ rounds, then, with probability at least $1/n^c$ after $T$ rounds, all ants are informed of the winning nest. We showed that with probability at least $1-1/n^c$, at least $6c\sqrt{n}$ ants are ignorant at the end of round $r = (\log_4 n)/2 - \log_4(12c)$. Therefore, $T \geq r = \Omega(\log n)$.
\end{proof}

\section{Optimal Algorithm}
\label{sec:optimal}
We present an algorithm that solves the \textsc{HouseHunting} problem and is asymptotically optimal. In the key step of the algorithm, each ant tries to recruit other ants to the nest it found after searching; after each round of recruiting, each ant checks if the number of ants at its nest has increased or decreased. Nests with a non-decreasing population continue competing while nests with a decreasing population drop out. In each round, the population of at least one nest is non-decreasing, so at least one nest will remain in the competition. Additionally, other nests drop out at a constant rate, meaning a single winning nest will be identified quickly.

This algorithm relies heavily on the synchrony in the execution and the precise counting of the number of ants at a given nest, which makes it sensitive to perturbations of these values, and therefore, not a natural algorithm that resembles ant behavior. However, the algorithm demonstrates that the \textsc{HouseHunting} problem is solvable in optimal time in the model of Section \ref{sec:model}.

\subsection{Algorithm Pseudocode and Description}

The pseudocode of the algorithm is presented in Algorithm \ref{algo:optimal}. Each call to the functions from Section \ref{sec:model} (in bold) takes exactly one round. The remaining lines of the algorithm are considered to be local computation and are executed in the same round as the preceding function call. 

Throughout the algorithm's execution, each ant is in one of four states: $search$, $active$, $passive$, or $final$, initially $search$. Based on the state of the ant, it executes the corresponding case block from the pseudocode. 
An ant is said to be \emph{committed} to a nest $n_i$ if $nest = i$.

\begin{algorithm}[h!]
\caption{Optimal \textsc{HouseHunting} Algorithm}	
\label{algo:optimal}
 \begin{multicols}{2}
$state:$ $\{search, active, passive,final\}$, \\ \hspace{12mm} initially $search$\\
$nest:$ nest index $i \in \{0,\cdots,k\}$, initially $0$ \\
$count:$ an integer in $\{0, \cdots, n\}$, initially $0$ \\
$quality:$ a boolean in $\{0,1\}$, initially $0$ \\
\vspace{2mm}
	\Case{$state = search$} {
		$<nest, count, quality> := \bv{search()}$ \textcolor{blue}{(R1)} \\
		\If{$quality = 0$} {
			$state := passive$
		}
		\Else{
			$state := active$
		}
	}
	\vspace{2mm}
	\Case{$state = passive$}{
		$\bv{go(nest)}$ \textcolor{blue}{(R1)} \\
		$<nest_t, \cdot> := \bv{recruit(0,nest)}$ \textcolor{blue}{(R2)} \\
		\If{$nest_t \neq nest$} {
			$nest := nest_t$ \\
			$state := final$
		}
		$\bv{go(nest)}$ \textcolor{blue}{(R3)} \\
		$\bv{go(nest)}$ \textcolor{blue}{(R4)}
	}
	\vspace{2mm}
	\Case{$state = final$}{
		$<nest,\cdot> := \bv{recruit(1,nest)}$ \textcolor{blue}{(R1)} 
	}
	\vspace{2mm}
	\Case{$state = active$}{
		$<nest_t,\cdot> := \bv{recruit(1,nest)}$ \textcolor{blue}{(R1)}  \\
		$count_t := \bv{go(nest_t)}$ \textcolor{blue}{(R2)} \\
		\If	{$(nest_t = nest)$ and \\ \hspace{4mm} $(count_t \geq count)$}{
			$count := count_t$ \\
			$\bv{go(nest)}$ \textcolor{blue}{(R3)} \\
			$<\cdot ,count_h> := \bv{recruit(0,nest)}$ \textcolor{blue}{(R4)}\\
			\If{$count_h = count$} {
				$state := final$ \\
			}
		}
		\ElseIf{$(nest_t = nest)$ and \\ \hspace{12mm} $(count_t < count)$}{
			$state := passive$ \\
			$\bv{recruit(0,nest)}$ \textcolor{blue}{(R3)} \\
			$\bv{go(nest)}$ \textcolor{blue}{(R4)}
		}
		\Else{
			$nest := nest_t$ \\
			$count_n := \bv{go(nest)}$ \textcolor{blue}{(R3)}\\
			\If{$count_n < count_t$}{
				$state := passive$
			}
			$\bv{go(nest)}$ \textcolor{blue}{(R4)}
		}
	}	
	\end{multicols}
	\vspace{2mm}
\end{algorithm}

Note that the $search$ subroutine is executed only once, during the first round, and the $final$ subroutine is what active ants do at the end of the execution to recruit all ants to the winning nest. The other two subroutines, $active$ and $passive$, represent the actions of active (recruiting) ants, and ants from bad/dropped-out nests, respectively. Each one of these subroutines takes exactly four rounds; each round is labeled as either R1, R2, R3, or R4 in the pseudocode. The subroutines are carefully scheduled in such a way that these two types of ants do not meet until the end of the competition process when a single winning nest remains; otherwise, the competition between ants from competing nests would be slowed down by ants from dropped-out nests. Therefore, the $active$ and $passive$ subroutines are padded with $recruit(0,\cdot)$ and $go(nest)$ calls to achieve such interleaving (lines 13, 18--19, 35--36, 42 are such padding rounds). 

\begin{itemize}

\item \textbf{Search (lines 6--11): } In the first round, each ant searches for a nest and commits to it. If the nest has quality of $0$, the ant moves to the $passive$ state; otherwise, it moves to the $active$ state.

\item \textbf{Passive (lines 12--19): } A passive ant is an ant committed to either a bad or a dropped-out nest. The ant spends a round at its (non-competing) nest, then it goes home to be recruited. This call to $recruit(0,nest)$ never coincides with a $recruit(1,nest)$ of an active ant, so a passive ant can only get recruited by an ant in the $final$ state calling $recruit(1,nest)$. Once successfully recruited, the passive ant moves to the $final$ state and commits to the new nest.

\item \textbf{Final (lines 20--21): } An ant in the final state is aware that a single winning nest remains, so it recruits to it in each round. This call to $recruit(1,nest)$ coincides with the call to $recruit(0,nest)$ of passive ants, so once a single nest remains, passive ants are recruited to it in every fourth round.

\item \textbf{Active (lines 22--42): } An active ant tries to recruit other ants to its competing nest by executing $recruit(1,nest)$. The ant then goes to the resulting nest to count the number of ants there. Based on the values of the resulting nest ($nest_t$) and count ($count_t$), we consider three cases:
\begin{itemize}
\item \emph{Case 1 (lines 25--31):} $nest_t$ is the same as the nest that the ant is committed to and the number of ants in that nest has not decreased; therefore, the nest remains competing. As a result, the ant updates the new count and spends an extra round at the nest that has a special purpose with respect to Cases 2 and 3 below. Finally, the ant checks if the number of ants at the home nest is the same as the number of ants at the committed nest; if this is the case, it means that all ants have been recruited to a single winning nest and the ant switches to the $final$ state.

\item \emph{Case 2 (lines 32--36):} $nest_t$ is the same as the nest that the ant is committed to but the number of ants has decreased; therefore the nest is about to drop out. As a result, the ant sets its state to $passive$ and spends a round at the home nest, which coincides with the round an active ant spends at the committed nest in Case 1. 

\item \emph{Case 3 (lines 37--42):} $nest_t$ is different from the nest that the ant is committed, which indicates the ant got recruited to another nest. Although it already knows the number of ants ($count_t$) at the new nest, the ant updates that count ($count_n$). The reason for this is to determine whether this new nest is about to compete or drop out. If $count_t=count_n$, the nest is competing because the active ants in Case 1 are spending the same round at the committed nest; if $count_t>count_n$, the nest is dropping out because the ants in Case 2 already determined a decrease in the number of ants and are spending this round at the home nest.
\end{itemize}
\end{itemize}

\subsection{Correctness Proof and Time Bound}
As written, Algorithm \ref{algo:optimal} never terminates; after all ants are committed to the same nest and in the $final$ state, they continue to recruit (each other) in every round. This issue can easily be handled if ants check whether the number of ants at the home nest is the same as the number of ants at the candidate nest.  However, for simplicity, we choose not to complicate the pseudocode and consider the algorithm to terminate once all ants have reached the $final$ state and, thus, committed to the same unique nest.

\paragraph{Proof Overview: } The correctness proof and time bound of Algorithm \ref{algo:optimal} are structured as follows. 
By defining a slightly different but equivalent recruitment process, we show, in Lemma \ref{lem:symm-nest}, that a competing nest is equally likely to continue competing and to drop out. Consequently, as we show in Lemma \ref{lem:sum-negative}, each competing nest has a constant probability of dropping out. We put these lemmas together in Theorem \ref{thm:optimal} to show that, with high probability, Algorithm \ref{algo:optimal} solves the \textsc{HouseHunting} problem in $\BO(\log n)$ rounds: $\BO(\log k)$ rounds to converge to a single nest and $\BO(\log n)$ rounds until all passive ants are recruited to it.

Let $R_3$ be the set of round numbers $r$ such that $r \in R_3$ iff $r \mod 4 = 3$; these rounds are exactly the rounds in which only ants committed to active nests are located at the home nest and try to recruit each other. Also, let $C(i,r)$ for each $r \in R_3$ and each nest $n_i$ denote the set of ants committed to nest $n_i$ and competing at the home nest. This implies that $|\cup_{i\in[1, k]} C(i,r)| = c(0,r)$ for $r \in R_3$.

Let random variable $X_r^a$, for each ant $a$ and each round $r \in R_3$, take on values $-1$, $0$ or $1$ as follows. If ant $a$ gets recruited by another ant in round $r$, then $X_r^a = -1$; if ant $a$ successfully recruits another ant, then $X_r^a = 1$; otherwise, $X_r^a = 0$. 






Let random variable $Y_r^i$ denote the change in the number of ants at nest $n_i$ after one round of recruiting, where $r \in R_3$:
	\begin{equation*}
		Y_r^i = \sum_{a \in C(i,r)} X_r^a.
	\end{equation*}
	
Informally speaking, the change in population of nest $n_i$ is simply the sum of identically distributed $\{-1,0,1\}$ random variables with mean $0$ that take on non-zero values with constant probability. Therefore, the sum of these variables is negative with constant probability. However, proving this fact requires a more rigorous argument because the $X_r^a$ variables are not independent.

As specified by the recruitment model, $Y_r^i$ is the result of a randomized recruitment process in which all ants are ordered by a random permutation. Then, in order of the permutation, ants choose uniformly at random other ants at the home nest to recruit. The random variables involved in this process are the random permutation $P$ as well as the set of random choices of ants. Consider the vector $\bv{o}_r$, of length $c(0,r)$,  where $\bv{o}_r(a) = -1$ if ant $a$ is chosen by another ant, $\bv{o}_r(a)=1$ if ant $a$ chooses an ant $a'$ such that $a'$ had not already chosen an ant and $a'$ was not already chosen by another ant, and $\bv{o}_r(a) = 0$ otherwise.


\begin{lemma}
\label{lem:symm-nest}
	Let nest $n_i$ be a competing nest in round $r \in R_3$. Then, $P[Y_r^i < 0] = P[Y_r^i > 0]$.
\end{lemma}
\begin{proof}
Random variable $Y_r^i$ is simply the sum of values of $\bv{o}_r(a)$ corresponding to ants $a \in C(i, r)$. Since $\bv{o}_r$ has exactly the same number of $-1$s and $1$s, it is possible to choose some (non-random) permutation $P'$ that swaps the locations of the $-1$s and $1$s in $\bv{o}_r$. Choosing a random permutation $P$ in the recruitment process and then applying the swapping permutation $P'$ negates the value of $Y_r^i$. Moreover, choosing a uniform random permutation $P$ and then applying this swapping permutation $P'$ still yields a uniform random permutation. So, $-Y_r^i$ is distributed according to the exact same distribution as $Y_r^i$. Therefore, $Y_r^i$ is symmetric around $0$ and $P[Y_r^i < 0] = P[Y_r^i > 0]$.
\end{proof}

\begin{lemma}
\label{lem:sum-negative}
	Let nest $n_i$ be a competing nest in round $r \in R_3$. If $|C(i,r)| < c(0,r)$, then $P[Y_r^i < 0] \geq 1/66$.
\end{lemma}

\begin{proof}
Let $P_1$ and $P_2$ be uniform random permutations. Consider using $P_1$ as the uniform permutation in in the recruitment process that determines $Y_r^i$. Permutation $P_2$ swaps the position of a fixed ant $a^*$ not committed to nest $n_i$ with the position of an ant $a$ chosen uniformly at random among the ants $C(i,r)$ (the ants committed to $n_i$). The existence of such an ant is guaranteed by the assumption that $|C(i,r)| < c(0,r)$.

By Lemma \ref{lem:recruit-succeed}, $P[X_r^{a^*} = 1] \geq 1/16$. Conditioning on $Y_r^i = 0$, at most $1/2$ the ants $a \in C(i,r)$ have $X_r^a = 1$.  Therefore, with probability at least $1/2$, the ant $a$ chosen uniformly at random by permutation $P_2$ has $X_r^a < 1$. Conditioning on $X_r^{a^*} = 1$ can only increase this probability. So with probability at least $(1/2)(1/16)$, applying the composition of $P_1$ and $P_2$ to compute $Y_r^i$ increases its value by at least $1$.  Since the composition of $P_1$ and $P_2$ is also a uniform random permutation, the distribution of $Y_r^i$ remains exactly the same as the case when only $P_1$ is applied. Therefore, 
\begin{eqnarray*}
P[Y_r^i = 0]  &\leq & 1 - \frac{1}{32} P[Y_r^i = 0] \\
P[Y_r^i = 0] &\leq & \frac{32}{33}.
\end{eqnarray*}

By Lemma \ref{lem:symm-nest}, $P[Y_r^i < 0] = P[Y_r^i > 0]$. Therefore, $P[Y_r^i < 0] \geq 1/66$.
\end{proof}

\begin{theorem}
\label{thm:optimal}
For any constant $c > 0$, with probability at least $1 - 1/n^c$, Algorithm \ref{algo:optimal} solves the \textsc{HouseHunting} problem in $\BO(\log n)$ rounds.
\end{theorem}
\begin{proof}
In the first round, all ants search for nests, so the expected number of ants located each good nest is $n/k$. Assuming $k \leq n/(12(c+1)\log n) = \BO(n/\log n)$, by a Chernoff bound, it follows that, with probability at least $1 - 1/n^{c+1}$, at least $n/(2k)$ ants visits each good nest.

Let $k_r$ be a random variable denoting the number of competing nests in round $r \in R_3$. Suppose $k_r > 1$, and so   $|C(i,r)| < c(0,r)$. By Lemma \ref{lem:sum-negative}, $P[Y_r^i < 0] \geq 1/66$ for each nest $n_i$ among the $k_r$ competing nests. Therefore, $\E[k_{r+4}] \leq (65/66) k_r$. Also, note that for any round $r$, $k_r \leq k$. For $r = \log_{66/65} k + (c+1) \log_{66/65} n = \BO(\log n)$, where $c>0$ is an arbitrary constant, it follows that $\E[k_r] = 1/n^{c+1}$. By a Markov bound, $P[k_r \geq 1] \leq 1/n^{c+1}$, so with probability at least $1 - 1/n^{c+1}$, $k_r \leq 1$. Note that in any given round, it is not possible for all nests to experience decrease in population, so for each round $r$, we have $k_r \geq 1$. So, overall, after $\BO(\log n)$ rounds, with probability $1 - 1/n^{c+1}$, there is exactly one competing nest.

After there is exactly one competing nest in some round $r \in R_3$, all ants committed to it switch to the $final$ state and start recruiting the passive ants during every fourth round. Each recruited ant also transitions to the final state, resulting in at most $\BO(\log n)$ rounds until all ants are committed to the winning nest.

Therefore, in total, with probability at least $1-1/n^c$ each nest is discovered by at least one ant and all ants are committed to a single good nest in $\BO(\log n)$ rounds.
\end{proof}

\section{Simple Algorithm}
\label{sec:natural}

We now give a very simple algorithm (Algorithm \ref{algo:simple}) that solves the \textsc{HouseHunting} problem  in $\BO(k \log n)$ rounds with high probability. The main idea of the algorithm is that all ants initially search for nests and those that find good nests simply continuously recruit to their nests with probability proportional to nest population in each round. 
Ants in larger nests recruit at higher rates, and eventually their populations swamp the populations of smaller nests. This process is similar to the well-known \emph{Polya's urn} model \cite{chung2003generalizations}.

\subsection{Algorithm Description and Pseudocode}

In each round of Algorithm \ref{algo:simple}, each ant can be in one of two states: $active$ or $passve$, initially starting in the $active$ state. In the first round of the algorithm, all ants search for nests; the ants that find good nests remain in the $active$ state, and the ants that find bad nests switch to the $passive$ state. Then, the algorithm proceeds in alternating rounds of recruitment by all ants at the home nest (either active $\bv{recruit(1, \cdot)}$ or passive $\bv{recruit(0, \cdot)}$), and population assessment at candidate nests. In each round of population assessment, each ant chooses to recruit actively in the next round with probability $count/n$, where $count$ is the assessed population at the candidate nest, and $n$ is the total number of ants. When a passive ant gets recruited to a nest, it commits to the new nest and becomes active again. When an active ant gets recruited to a different nest, it commits to the new nest and starts recruiting for that new nest.

\begin{algorithm}
\caption{Simple House-Hunting}
\label{algo:simple}
 \begin{multicols}{2}
$state:$ $\{active, passive\}$, initially $active$ \\
$<nest, count, quality> := \bv{search()}$ \\
\If{$quality = 0$} {
			$state := passive$
		}
\Case{$state = active$} {
				$b := 1$ with probability $\frac{count}{n}$, $0$ otherwise \\
				$nest := \bv{recruit(b, nest)}$  \\
			$count:= \bv{go(nest)}$ \\
	}
\Case{$state = passive$} {
		$nest_t := \bv{recruit(0,nest)}$ \\
		\If{$nest_t \neq nest$} {
			$state := active$ \\
			$nest := nest_t$ \\
			$count:= \bv{go(nest)}$ \\
		}
	}
\end{multicols}
\vspace{3mm}
\end{algorithm}

\subsection{Correctness Proof and Time Bound}
For each nest $n_i$ and each round $r$ let random variable $p(i,r) =c(i,r)/n$ denote the proportion of ants at nest $n_i$ in round $r$. By the pseudocode, ants located at nest $n_i$ in round $r$ will recruit with probability $p(i,r)$ in round $r+1$. Define $\Sigma^2(r) = \sum_{i=1}^{k} p(i,r)^2$, the expected proportion of ants that will recruit in total, in round $r$. Since $\sum_{i=1}^{k} p(i,r) = 1$, by the $\ell_1$ versus $\ell_2$ norm bound $\Sigma^2(r) \ge 1/k$. 

\paragraph{Proof Overview: } The correctness proof and time bound of the algorithm are structured as follows. First, in Lemmas \ref{lem:exp-x}, \ref{lem:switch-constants}, and \ref{lem:expected_change}, we show some basic bounds on the expected number of ants recruited in each round and the change of the value of $p(i,r)$ for a single nest in each round of recruiting. Then, in Lemmas \ref{lem:taylor_math}, \ref{lem:taylor-applied}, and \ref{lem:taylor_series}, we use a Taylor series expansion of the ratio between the populations of two nests to show that in expectation this ratio increases multiplicatively by $(1+\Omega(1/k))$ in each round of recruiting, provided that both nests have a $\Omega(1/k)$ fraction of the total population. On the other hand, if a nest has less than a $\Omega(1/k)$ fraction of the total population, in Lemmas \ref{lem:small_nest} and \ref{lem:small_nest_dropout}, we show that the ants in such a nest recruit so slowly that the nest completely drops out within $\BO(k \log n)$ rounds with high probability. Finally, in Theorem \ref{thm:simple_runtime}, we consider all $\binom{k}{2}$ pairs of nests to show that only a single nest will remain in each pair within $\BO(k \log n)$ rounds, ensuring, by a union bound, that a single nest remains overall.

Throughout this section, let $c > 0$ be an arbitrary constant and let $d$ be an arbitrary constant such that $d \geq 64$. For the sake of analysis, assume that $k \leq \sqrt{n/(8 d^2 (c+6) \log n)}  = \BO(\sqrt{n/\log n})$. While we feel that this assumption is reasonable, we are also hopeful that it could be removed.
Let $R_1$ be the set of all odd rounds, excluding round $1$; by the pseudocode, in these rounds, ants are located at the candidate nests. Therefore, for each $r \in R_1$ and each ant $a$, $\ell(a,r) = i \neq 0$; for each round $r' \not\in R_1$ and for each ant $a$, $\ell(a,r') = 0$.

\subsubsection{Change in Population of a Single Nest in One Round}
We first study how the population of a single nest changes in a single round. Intuitively, we expect a $p(i,r)$ proportion of the ants in $n_i$ to recruit, and a $\Sigma^2(r)$ proportion to be recruited. Therefore, we expect the population to change by $p(i,r)\left [p(i,r) - \Sigma^2(r) \right ]$. Qualitatively, we show this in Lemma \ref{lem:expected_change}, modulo constant factors; the main technical difficulty is handling dependencies in the random recruiting process. Before we are ready to prove Lemma \ref{lem:expected_change}, we show two lemmas that state the expected outcome of a single ant recruiting. 

As before, define random variable $X_r^a$ to take on value $-1$ if ant $a$ is recruited away from its current nest in round $r$, $1$ if it successfully recruits another ant, and $0$ otherwise. 

Note that the following lemma is stated in a slightly generalized way: it applies to an ant $a$ recruiting with a fixed probability $p$ in a given round $r$. By the pseudocode of the algorithm it is clear that whenever an ant recruits, it is always with probability $p(i,r)$, so the majority of the time we will apply the lemma for $p=p(i,r)$; however, the general statement of the lemma helps reason about the expected value of $X_r^a$ with respect to two ants from two different nests in the proof of Lemma \ref{lem:switch-constants}. 

\begin{lemma}
\label{lem:exp-x}
	Let $n_i$ be any nest and let ant $a$ be located in nest $n_i$ in some round $r \in R_1$. Also, suppose ant $a$ recruits with probability $p$ in round $r$. Then, there exist functions $\xi_1$ and $\xi_2$ such that $\E [X_r^a] = p \xi_1(i,r) - \xi_2(i,r)$ and $\xi_1(i,r) \ge \xi$ for a fixed constant $\xi > 0$.
\end{lemma}

\begin{proof}
Let $A_r^a$ be an indicator variable indicating whether some ant $a$ chooses to recruit (executes $\bv{recruit(1,\cdot)}$) in round $r+1$. By assumption, $\Pr \left [A_r^a = 1 \right] = p$.
Whether ant $a$ is actually successful in recruiting another ant depends on: (1) the order of the random recruiting permutation $P$, (2) the choices of other ants to recruit or not, and (3) the choices by recruiting ants of whom to recruit. Let $B$ be a random variable encompassing all these random variables that affect $X_r^a$, excluding $A_r^a$. Therefore, $B$ is a triple $(P, \{A_r^{a'} | a' \neq a\}, \{1, \cdots, n\}^n)$ and it takes on values from the set $\mathcal{B}$ of all such triples. The expected value of $X_r^a$ is:

\begin{equation*}
\E [X_r^a] = \sum_{b \in \mathcal B} \Pr \left [B=b \right ] \cdot \left ( p \E \left [ X_r^a | B=b, A_r^a = 1 \right ] + (1-p) \E \left [ X_r^a | B=b,A_r^a = 0 \right ] \right ).
\end{equation*}

Fix some value $B = b$. We consider several cases based on the nine possible pairs of values of $A_r^a$ and $X_r^a$ ($A_r^a$ has two possible values and $X_r^a$ has three possible values). However, note that some permutations of these values are not allowed; that is, if $A_r^a = 0$ (the ant chooses not to recruit), it is not possible that $X_r^a=1$ (the ant succeeds in recruiting). Also, since $B=b$ is already fixed, ant $a$ is either chosen by another ant or not, regardless of the value of $A_r^a$. This fact rules out two more cases: (1) the case where $X_r^a = 0$ if $A_r^a = 0$, and $X_r^a = -1$ if $A_r^a = 1$, and (2) the case where $X_r^a = -1$ if $A_r^a = 0$, and $X_r^a = 0$ if $A_r^a = 1$. In (1), since the choices of the other ants are already fixed and included in $B=b$, it is not possible that ant $a$ gets recruited if it chooses to recruit but it does not get recruited otherwise. Similarly, in (2), it is not possible that ant $a$ gets recruited when it chooses not to recruit but does not get recruited when it fails to recruit.

All remaining cases are listed below:

\begin{enumerate}
\item[Case 1:] $X_r^a = -1$ for both $A_r^a = 0$ and $A_r^a = 1$. That is, the ant is recruited by another ant no matter its decision.
\item[Case 2:] $X_r^a = -1$ if $A_r^a = 0$, and $X_r^a = 1$ if $A_r^a = 1$. That is, if the ant chooses to recruit it succeeds and if not, it is recruited by another ant.
\item[Case 3:] $X_r^a = 0$ for both $A_r^a = 0$ and $A_r^a = 1$. That is, whether the ant chooses to recruit or not, it will not be part of a successful recruitment. 
\item[Case 4:] $X_r^a = 0$ if $A_r^a = 0$, and $X_r^a = 1$ if $A_r^a = 1$. That is, if the ant chooses to recruit it succeeds, but if not, it is not recruited by another ant.
\end{enumerate}

Let $p_c^b$ be the probability of case $c$ occurring given $B=b$. The expected value of $X_r^a$ is:
\begin{eqnarray*}
\E [X_r^a] &=& \sum_{b \in \mathcal B} \Pr \left [B=b \right ]  \cdot \left(p \left ( p_2^b + p_4^b - p_1^b \right ) + (1-p) \left ( -p_1^b-p_2^b \right ) \right)\\
&=& \sum_{b \in \mathcal B} \Pr \left [B=b \right ] \left (p \left (2p_2^b + p_4^b \right ) - \left ( p_1^b + p_2^b \right ) \right )\\
& =& p \E_B \left [2p_2^b + p_4^b \right ] - \E_B \left [ p_1^b + p_3^b \right ].
\end{eqnarray*}

Now, we need to show that $ \E_B \left [2p_2^b + p_4^b \right ] \ge \xi $ for some fixed constant $\xi > 0$.  
By Lemma \ref{lem:recruit-succeed}, each recruiting ant succeeds with some probability lower bounded by a constant. A recruiting ant only succeeds in cases $2$ and $4$ so we have $\E_B [p_2^b + p_4^b] \ge \xi$ for some constant $\xi > 0$. Therefore, $\E_B [2p_2^b + p_4^b] \ge\E_B [p_2^b + p_4^b] \ge \xi$.

Overall, $\E [X_r^a] = p \xi_1(i,r) - \xi_2(i,r)$, where $\xi_1(i,r) = \E_B [ 2p_2^b + p_4^b]$ is any constant greater than $\xi>0$ and $\xi_2(i,r) = \E_B[ p_1^b + p_2^b]$ is some unspecified value.
\end{proof}

Note that the values of $\xi_1(i,r)$ and $\xi_2(i,r)$ are slightly different for different nests because the behavior of an ant depends on the nest it is committed to. For example, the different probabilities corresponding to the cases in Lemma \ref{lem:exp-x} vary based on the population of a given nest in a given round and the populations of the remaining nests in that round. Next, we prove a bound on the $\xi_1(i,r)$ and $\xi_2(i,r)$ values of two different nests.

Fix a constant $\xi$ and functions $\xi_1$ and $\xi_2$ such that $\xi_1(i,r) \geq \xi$ for all rounds $r$ and all nests $n_i$.

\begin{lemma}
\label{lem:switch-constants}
Let $n_i$ and $n_j$ be two nests with $p(i,r) \le p(j,r)$ in some round $r \in R_1$. Then, $\xi_1(i,r) \cdot p(i,r)- \xi_2(i,r) \le \xi_1(j,r)\cdot p(i,r) - \xi_2(j,r)$.
\end{lemma}
\begin{proof}
Let $a_i$ and $a_j$ be two ants located in nests $n_i$ and $n_j$, respectively, in round $r \in R_1$. Ant $a_i$ recruits with probability $p(i,r)$ and ant $a_j$ recruits with probability $p(j,r)$ where $p(i,r) \le p(j,r)$. By Lemma \ref{lem:exp-x}, applied to ant $a_i$, recruiting with probability $p(i,r)$, and nest $n_i$, it follows that $\E[X_r^{a_i}] = p(i,r) \xi_1(i,r) - \xi_2(i,r)$.

 Suppose we swap the locations of ants $a_i$ and $a_j$ and let both ants recruit with probability $p(i,r)$; since the ants just switch locations, the populations of the two nests remain the same, so the rest of the ants still recruit with probability $p(i,r)$ in nest $n_i$ and $p(j,r)$ in nest $n_j$\footnote{Here it is not important to which nests ants are recruiting and bringing other ants; we just want to reason about the value of $X_r^{a_i}$. Therefore, it is not alarming that ant $a_i$ is located at nest $n_j$ and recruiting with probability $p(i,r)$.}. By Lemma \ref{lem:exp-x}, applied to ant $a_i$, recruiting with probability $p(i,r)$, and nest $n_j$, it follows that the expected value of $X_r^{a_i}$ in this hypothetical situation is $\E'[X_r^{a_i}] = p(i,r) \xi_1(j,r) - \xi_2(j,r)$. 

As a result of switching the positions of ants $a_i$ and $a_j$, we are essentially lowering $a_j$'s probability to recruit. Therefore, the probability that ant $a_i$ succeeds recruiting, if it chooses to recruit, increases, and the probability that ant $a_i$ is recruited, if it chooses not to recruit, decreases. 

Therefore, $\E[X_r^{a_i}] \leq \E'[X_r^{a_i}]$, and so $p(i,r)\xi_1(i,r) - \xi_2(i,r) \le p(i,r)\xi_1(j,r) - \xi_2(j,r)$.
\end{proof}

Next, we use Lemma \ref{lem:exp-x} to calculate the expected change in the value of $p(i,r)$ after one round of recruiting.

\begin{lemma}\label{lem:expected_change}
For each nest $n_i$ and each round $r \in R_1$:
\begin{equation*}
 	\E [p(i,r+2)] = p(i,r)\cdot \left [ 1+ \xi_1(i,r) \cdot p(i,r) - \xi_2(i,r) \right ].
 \end{equation*}
\end{lemma}
\begin{proof} 
 By linearity of expectation, 
\begin{equation*}
\E [p(i,r+2)] = p(i,r) + \frac{1}{n} \sum_{\{a | \ell(a,r) = i\}} \E [X_r^a] = p(i,r) \cdot \left (1+\E [X_r^a]  \right ),
\end{equation*}

where in the last equation $a$ is a fixed ant in nest $n_i$ and the equality follows from the fact that $X_r^a$ is identically distributed for each ant in the same nest. 

By Lemma \ref{lem:exp-x}, $\E [X_r^a] = p(i,r) \xi_1(i,r) - \xi_2(i,r)$, so:

\begin{equation*}
\E [p(i,r+2)] = p(i,r) \cdot \left (1+p(i,r) \xi_1(i,r) - \xi_2(i,r)  \right ). \qedhere
\end{equation*}
\end{proof}

Lemma \ref{lem:expected_change} shows that the population change in a single nest depends \emph{quadratically} on $p(i,r)$. With this fact, we can show that larger nests will tend to `swamp' smaller nests, causing their populations to drop to $0$. 

\subsubsection{Relative Changes in the Populations of Two Nests in One Round}

We first define a measurement of the population gap between two nests.

\begin{definition}
For any two nests $n_i$ and $n_j$, let $n_H(i,j,r) \in \{i, j \}$ be the id of the nest with the higher population in round $r$. Let $n_L(i,j,r)$ be the nest with the lower population. For simplicity of notation let $p_H(i,j,r) = p(n_H(i,j,r),r)$ and $p_L(i,j,r) = p(n_H(i,j,r),r)$. Define:
\begin{align*}
\epsilon(i,j,r) = \frac{p_H(i,j,r)}{p_L(i,j,r)} - 1.
\end{align*}
\end{definition}

That is, $\epsilon(i,j,r)$ is the relative population gap between the larger and smaller of nests $n_i$ and $n_j$ in round $r$. 

\begin{lemma}
\label{lem:eps-bound}
	For any two nests $n_i$ and $n_j$, $\E [\epsilon(i,j, 1)] \geq 1/(3(n-1))$.
\end{lemma}
\begin{proof}
We need to bound the expected difference between two nests after the first round of searching. If nests $n_i$ and $n_j$ receive the same number of ants after the initial round of searching, then, by definition, $\epsilon(i,j, 1) = 0$. Let $2 \leq x \leq n$ be the total number of ants that find either nest $n_i$ or nest $n_j$ in the first round of searching. Note that if $x=1$, it is not possible for both nests to receive the same number of ants, and if $x=0$, then the nests are already empty. The probability that both nests receive exactly $x/2$ ants each is:
\begin{equation*}
	\Pr[\epsilon(i,j, 1) = 0] = \binom{x}{x/2} \left(\frac{1}{2}\right)^{x/2} \left(\frac{1}{2}\right)^{x/2} \leq \left(\frac{e2^x}{\pi \sqrt{x}}\right) \frac{1}{2^x} \leq \frac{e}{\pi \sqrt{x}} \leq \frac{e}{\pi \sqrt{2}} < \frac{2}{3}, 
\end{equation*}
where in the second step we use the Stirling approximation $\sqrt{2\pi} x^{x+1/2} e^{-x} \leq x! \leq e x^{x+1/2} e^{-x}$. Conditioning on $\epsilon(i,j, 1) \neq 0$, the smallest possible value of $\epsilon(i,j, 1)$ is $1/(n-1)$, and so $\E [\epsilon(i,j, 1)|\epsilon(i,j, 1) \neq 0] \geq 1/(n-1)$. Therefore, 
\begin{equation*}
\E [\epsilon(i,j, 1)] \geq \E [\epsilon(i,j, 1)|\epsilon(i,j, 1) \neq 0] \cdot \Pr[\epsilon(i,j, 1) \neq 0] \geq \frac{1}{3(n-1)}. \qedhere
\end{equation*}
\end{proof}

Since $\epsilon(i,j,r)$ is a ratio of two populations, it is not immediately clear how to compute its expected change between rounds. However, using a Taylor series expansion we are able to \emph{linearize} this ratio, and then use Lemma \ref{lem:expected_change} to compute its expected change. We first show how to use the Taylor series expansion in our setting.

\begin{lemma}\label{lem:taylor_math}
Consider positive numbers $x_0$, $x_1$, $y_0$, and $y_1$. Let $\Delta x = x_1 - x_0$, $\Delta y = y_1 - y_0$, and $\Delta \left ( \frac{x}{y} \right ) = \frac{x_1}{y_1} - \frac{x_0}{y_0}$. Then:
\begin{align}\label{eq:main_taylor}
\Delta \left ( \frac{x}{y} \right ) = \left [\frac{\Delta x}{y_0}- \frac{x_0\Delta y}{y_0^2} \right] \cdot  \sum_{i=0}^\infty \left (\frac{-\Delta y}{y_0} \right )^i
\end{align}
\end{lemma}

\begin{proof}
We can expand $\Delta \left ( \frac{x}{y} \right )$ using a Taylor series for $f(x,y) = \frac{x}{y}$. To save notation, we write $\frac{\partial^k f}{\partial x^i \partial y^{k-i}}$ to denote $\frac{\partial^k f}{\partial x^i \partial y^{k-i}}\left(x_0,y_0 \right )$. We have:
\begin{align*}
\Delta \left ( \frac{x}{y} \right ) = \frac{\partial f}{\partial x} \Delta x +  \frac{\partial f}{\partial y} \Delta y + \frac{1}{2!} \left [\frac{\partial^2 f}{\partial x^2} \Delta x^2 +  2\frac{\partial^2 f}{\partial x \partial y} \Delta x \Delta y + \frac{\partial^2 f}{\partial y^2} \Delta y^2 \right] + \frac{1}{3!} \left [ \frac{\partial^3 f}{\partial x^3} \Delta x^3 + ... \right ] + ...
\end{align*}
We will show by induction that, for all $k$:
\begin{align}\label{eq:taylor_induction}
\frac{1}{k!} \left [ \frac{\partial^k f}{\partial x^k} \Delta x^k + k  \frac{\partial^k f}{\partial x^{k-1}\partial y} \Delta x^{k-1} \Delta y + ... +   \frac{\partial^k f}{\partial y^k} \Delta y^k  \right ] =  \left [\frac{\Delta x}{y_0}- \frac{x_0\Delta y}{y_0^2} \right] \left [ \frac{-\Delta y}{y_0} \right ]^{k-1}.
\end{align}
Plugging into the Taylor series expansion, this immediately gives \eqref{eq:main_taylor}. 
For $k = 1$, we differentiate $f(x,y) = \frac{x}{y}$ to get our base case:
\begin{align*}
 \frac{\partial f}{\partial x} \Delta x +  \frac{\partial f}{\partial y} \Delta y = \frac{\Delta x}{y_0}- \frac{x_0\Delta y}{y_0^2}.
\end{align*}

Now for $k > 1$, by the inductive assumption we have:
\begin{align*}
\frac{1}{(k-1)!} \left [ \frac{\partial^{k-1} f}{\partial x^{k-1}} \Delta x^{k-1} + ... +   \frac{\partial^{k-1} f}{\partial y^{k-1}} \Delta y^{k-1}  \right ] =  \left [\frac{\Delta x}{y_0}- \frac{x_0\Delta y}{y_0^2} \right] \left [ \frac{-\Delta y}{y_0} \right ]^{k-2} = \frac{\Delta x(-\Delta y)^{k-2}}{y_0^{k-1}}+ \frac{x_0(-\Delta y)^{k-1}}{y_0^k}.
\end{align*}
 Differentiating gives:
\begin{align*}
\frac{1}{k!} \left [ \frac{\partial^k f}{\partial x^k} \Delta x^k + ... +   \frac{\partial^k f}{\partial y^k} \Delta y^k  \right ] &= \frac{1}{k} \left [ \frac{(-\Delta y)^{k-1}}{y_0^k}\Delta x - \frac{(k-1)\Delta x (-\Delta y)^{k-2}}{y_0^k} \Delta y - \frac{k x_0 (-\Delta y)^{k-1}}{y_0^{k+1}} \Delta y  \right]\\
&=  \left [ \frac{\Delta x \cdot (-\Delta y)^{k-1}}{y_0^k} + \frac{x_0 \Delta x (-\Delta y)^{k}}{y_0^{k+1}}  \right] =  \left [\frac{\Delta x}{y_0}- \frac{x_0\Delta y}{y_0^2} \right] \left [ \frac{-\Delta y}{y_0} \right ]^{k-1}.
\end{align*}
\end{proof}

For simplicity of notation, write $p_H =  p(n_H(i,j,r),r)$, $p_L =  p(n_L(i,j,r),r)$. Additionally, write $\Delta p_H = p(n_H(i,j,r),r+2) - p(n_H(i,j,r),r)$, $\Delta p_L = p(n_L(i,j,r),r+2) - p(n_L(i,j,r),r)$ and $\Delta \left (\frac{p_H}{p_L} \right ) = \frac{p(n_H(i,j,r),r+2)}{p(n_L(i,j,r),r+2)} - \frac{p(n_H(i,j,r),r)}{p(n_L(i,j,r),r)}$. We can now apply Lemma \ref{lem:taylor_math} to show:

\begin{lemma}
\label{lem:taylor-applied}
	Let $n_i$ and $n_j$ be two nests with $\Delta p_L < p_L$ for some round $r \in R_1$. Then, $\Delta \left (\frac{p_H}{p_L} \right) \ge \frac{1}{2} \left( \frac{\Delta p_H}{p_L}- \frac{p_H\Delta p_L}{p_L^2} \right)$.
\end{lemma}

\begin{proof}
	 Note that we are not considering expectations for now, just the actual values. Applying Lemma \ref{lem:taylor_math}, we can write $\Delta \left (\frac{p_H}{p_L} \right )$ in terms of $\Delta p_H$ and $\Delta p_L$. We have:
\begin{align*}
\Delta \left (\frac{p_H}{p_L} \right) = \left( \frac{\Delta p_H}{p_L}- \frac{p_H\Delta p_L}{p_L^2} \right)  \left( 1- \left (\frac{\Delta p_L}{p_L} \right ) + \left (\frac{\Delta p_L}{p_L} \right ) ^2 - \left (\frac{\Delta p_L}{p_L}\right ) ^3 +... \right).
\end{align*}

The value $\frac{\Delta p_L}{p_L}$ is strictly less than $1$ by the assumption that $\Delta p_L < p_L$.
So:
\begin{equation*}
	\sum_{i=0}^{\infty} \left( -\frac{\Delta p_L}{p_L} \right)^i = \frac{1}{1+ \Delta p_L/p_L} \ge \frac{1}{1+1} = 1/2. \qedhere
\end{equation*}
\end{proof}

Next, we use the value of $\Delta(p_H/p_L)$ in order to calculate the expected change in $\epsilon(i,j,r)$ for two nests $n_i$ and $n_j$ after one round of recruiting.

\begin{lemma}\label{lem:taylor_series}
Let $n_i$ and $n_j$ be two nests with $\Delta p_L < p_L $, $p(i,r) \geq 1/(dk)$ and  $p(j,r) \geq 1/(dk)$ for some round $r \in R_1$. Then, $\E \left [ \epsilon(i,j,r+2) \right ] \ge \left (1+1/(2dk) \right ) \E \left [ \epsilon(i,j,r) \right ]$.
\end{lemma}
\begin{proof}
First we show that:

\begin{equation*}
	\epsilon(i,j,r+2) \ge  \frac{p(n_H(i,j,r),r+2)}{p(n_L(i,j,r),r+2)} - 1.
\end{equation*} 

This is true because, if the larger of the two nests in round $r$ remains larger in round $r+2$, then $n_H(i,j,r) = n_H(i,j,r+2)$ and $n_L(i,j,r) = n_L(i,j,r+2)$, so, by the definition of $\epsilon$, $\epsilon(i,j,r+2) =  \frac{p(n_H(i,j,r),r+2)}{p(n_L(i,j,r),r+2)} - 1$. On the other hand, if the two nests flip positions, then  $\frac{p(n_H(i,j,r),r+2)}{p(n_L(i,j,r),r+2)} - 1$ must be smaller than $\epsilon(i,j,r+2)$ since it contains the population of the larger nest in the denominator of the fraction. Fixing $p(i,r)$ and $p(j,r)$, we have: 

\begin{eqnarray*}
	\E [\epsilon(i,j,r+2)] \ge \E  \left[ \frac{p(n_H(i,j,r),r+2)}{p(n_L(i,j,r),r+2)} - 1\right] &=& \epsilon(i,j,r) + \E \left[\frac{p(n_H(i,j,r),r+2)}{p(n_L(i,j,r),r+2)} - 1 - \epsilon(i,j,r) \right]\\
&=& \epsilon(i,j,r) + \E \left[\frac{p(n_H(i,j,r),r+2)}{p(n_L(i,j,r),r+2)} - \frac{p(n_H(i,j,r),r)}{p(n_L(i,j,r),r)} \right]
\end{eqnarray*}

So we have reduced our question of how we expect the \emph{absolute difference} between the nests to change  in a single round to the question of how we expect the ratio of population to change.

By Lemma \ref{lem:taylor-applied}, it follows:
\begin{equation*}
\E [\epsilon(i,j,r+2)] \geq \epsilon(i,j,r) + \frac{1}{2} \left (\E \left[ \frac{\Delta p_H}{p_L}- \frac{p_H\Delta p_L}{p_L^2} \right] \right ) \geq \epsilon(i,j,r) + \frac{1}{2}  \left( \frac{p_L \E[\Delta p_H] - p_H\E[\Delta p_L]}{p_L^2}\right ).
\end{equation*}

By Lemma \ref{lem:expected_change} we have:
\begin{eqnarray*}
p_L \E[\Delta p_H] - p_H\E[\Delta p_L] &=& p_Lp_H^2 \xi_1(n_H,r)  - p_Lp_H\xi_2(n_H,r) - p_H p_L^2 \xi_1(n_L,r)  + p_L p_H\xi_2(n_L,r)\\
&=& p_Lp_H \left [ p_H \xi_1(n_H,r)  - \xi_2(n_H,r) - p_L \xi_1(n_L,r)  + \xi_2(n_L,r) \right ]\\
&\ge &  p_Lp_H \xi_1(n_L,r) \cdot (p_H-p_L) \ge  \xi p_Lp_H (p_H-p_L),
\end{eqnarray*}

where the inequality follows from the guarantee of Lemma \ref{lem:switch-constants} that $p_L\xi_1(n_L,r) - \xi_2(n_L,r) \le p_L\xi_1(n_H,r) - \xi_2(n_H,r)$.
Noting that $\epsilon(i,j,r) = \frac{p_H-p_L}{p_L} $ we have:
\begin{equation*}
\E [\epsilon(i,j,r+2)] \geq \epsilon(i,j,r) + \frac{1}{2} \left (  \frac{p_Lp_H(p_H-p_L)}{p_L^2} \right ) \geq \epsilon(i,j,r) \left(1 + \frac{p_H}{2} \right) \geq \epsilon(i,j,r)\left( 1 + \frac{1}{2dk} \right). \qedhere
\end{equation*}
\end{proof}

\subsubsection{Changes in Populations of Nests over O(k log n) Rounds}

We now show that once the population of a nest is very small, it will quickly drop to zero.

\begin{lemma}\label{lem:small_nest}
Let $n_i$ be any nest with $p(i,r) \le 1/(dk)$ in some round $r \in R_1$. Then, with probability at least $1 - 1/n^{c+4}$, $p(i,r') \leq 1/(dk)$ for all rounds $r' \in [r, r + 64(c+4)k \log n)]$.
\end{lemma}
\begin{proof}
 It suffices to show that if $p(i,r) \le 1/(dk)$ then $p(i,r+2) \le 1/(dk)$ with probability $1-1/n^{c+5}$. We can then union bound over all $O(k \log n)$ rounds to get that, with  probability $1 - 1/n^{c+4}$,  $p(i, r') \leq 1/(dk)$ for all rounds $r' \in [r, r + 64(c+4)k \log n)]$.

Consider two possible cases:
\begin{itemize}
	\item[Case 1: ] $p(i,r) < 1/(2dk)$. Then, even if all ants successfully recruit, the nest cannot more than double in size, so $p(i,r+2) < 1/(dk)$.
	\item[Case 2: ] $p(i,r) \geq 1/(2dk)$. The expected number of ants to recruit to nest $n_i$ is $n \cdot p(i,r)^2 \ge n/(4d^2 k^2)$. By the assumption that $k \leq \sqrt{n/(8 d^2 (c+6) \log n)}  = \BO(\sqrt{n/\log n})$, this value is in $\Omega( \log n)$. By a Chernoff bound, with probability at least $1-1/n^{c+6}$, at most $2n \cdot p(i,r)^2 = n/(2d^2k^2)$ ants choose to recruit.
	
	Similarly, the expected the number of ants to be recruited from nest $n_i$ is: 
\begin{equation*} 
 n  p(i,r) \left(1-p(i,r)\right ) \Sigma^2(r) \geq n \frac{1}{dk} \left( 1 - \frac{1}{2dk} \right) \frac{1}{k}  \geq \frac{n}{dk^2} \left( 1 - \frac{1}{2dk} \right), 
\end{equation*}
	where in the second step we use the fact that $1/(dk) > p(i,r) \geq 1/(2dk)$ and $\Sigma^2(r) > 1/k$. This value is also in $\Omega (\log n)$.
	
Next, we would like to apply a Chernoff bound to the number of recruited ants from nest $n_i$; however, recruitments are not independent, so we need to do a simple domination trick, similarly to the proof of the lower bound. 

Let random variable $Y_r^a$ have value $1$ with probability $1/k(1-1/(2dk))$, and let $Y_r$ be the sum of $n/(dk)$ such independent random variables. So, $\E[Y_r] = (n/(dk^2))(1-1/(2dk))$, and by a Chernoff bound, $P[Y_r < (3/2)(n/(dk^2)(1-1/(2dk))] < n^{-(c+6)}$. We can see that:

\begin{equation*}
 P[a \text{ is recruited } | \text{ the choices of all other ants}] \geq (1-p(i,r)) \Sigma^2(r) \geq \left( 1 - \frac{1}{2dk} \right) \frac{1}{k} = P[Y_r^a=1]
\end{equation*} 

Therefore, by Lemma \ref{lem:dependent-chernoff}, $P[\# \text{ of ants recruited from nest }n_i < x] < P[Y_r < x]$ for some $x \leq n/(dk)$. Plugging in $x = (3/2)(n/(dk^2)(1-1/(2dk))$, it follows that:
\begin{equation*}
	P[\text{\# of ants recruited from nest }n_i < (3/2)(n/(dk^2)(1-1/(2dk))] < n^{-(c+6)}.
\end{equation*}

Therefore, with probability at least $1 - 1/n^{c+6}$, at least $(3/2)(n/(dk^2)(1-1/(2dk))$ ants will be recruited from nest $n_i$. Recall that, with probability at least $1 - 1/n^{c+6}$, at most $n/(2d^2k^2)$ ants from nest $n_i$ choose to recruit. Since $d \geq 64$, it is easy to see that the number of ants recruiting is smaller than the number of ants recruited away from nest $n_i$, so with probability at least $1 - 1/n^{c+5}$, the total population of $n_i$ decreases, indicating $p(i,r+2) < 1/(dk)$.
\end{itemize}

So overall, if $p(i,r) < 1/(dk)$, then, with probability at least $1 - 1/n^{c+4}$, $p(i,r') < 1/(dk)$ for all $r' \in [r, r + \BO(k \log n)]$. 
\end{proof}

\begin{lemma}\label{lem:small_nest_dropout}
Let $n_i$ be a nest with $p(i,r) \le 1/(dk)$ in some round $r \in R_1$. Then, with probability at least $1 - 1/n^{c+3}$, $c(i,r') = 0$ for $r' = 64(c+4)k \log n)$. 
\end{lemma}

\begin{proof}
We calculate the expected change in the number of ants in nest $n_i$ by first calculating $\E[X_r^a]$ for some ant $a$ in nest $n_i$ and some round $r \in R_1$ such that $p(i,r) \le 1/(dk)$. 

Suppose ant $a$ does not recruit. Let random variable $Y_r^{a'}$ have value $1$ if ant $a'$ successfully recruits ant $a$, and value $0$ otherwise. By Lemma \ref{lem:recruit-succeed}, it follows that $\E[Y_r^{a'}] \geq p(i',r)/(16n)$ where $i' = \ell(a',r)$. Therefore:

\begin{eqnarray*}
	\E\left[\sum_{a' \neq a} Y_r^{a'} \right] &=& \sum_{a' \neq a} \E[Y_r^{a'}] \geq \frac{1}{16n} \sum_{j=1}^k \sum_{\{a' \neq a | \ell(a',r)=j\}} p(j,r) \\
	 &\geq & \frac{1}{16n} \left( \sum_{j=1}^k \sum_{\{a' | \ell(a',r)=j\}} p(j,r)\right) - \frac{1}{dk} \\
	 &\geq & \frac{1}{16n} \left( \sum_{j=1}^k n p(j,r)^2 \right) - \frac{1}{dk} \geq \frac{\Sigma^2(r)}{16} - \frac{1}{dk}
\end{eqnarray*}	

Therefore, the expected value of $X_r^a$ is:

\begin{equation*}
\E [X_r^a] \leq \frac{1}{dk}-\left (1-\frac{1}{dk} \right) \E\left[\sum_{\{a' | \ell(a',r)\neq i\}} Y_r^{a'}\right] = 
\frac{1}{dk}-\left (1-\frac{1}{dk} \right) \left( \frac{\Sigma^2(r)}{16} - \frac{1}{dk} \right) \le -\frac{1}{64k},
\end{equation*}
where the last step follows from the assumption that $d \geq 64$. Therefore, 
\begin{equation*}
\E [p(i,r+2)] = p(i,r) + \frac{1}{n} \sum_{\{a | \ell(a,r) = i\}} \E [X_r^a]  \leq p(i,r)(1+\E[X_r^a]) \leq p(i,r) \left (1-\frac{1}{64k}\right)
\end{equation*}

By Lemma \ref{lem:small_nest}, $p(i,r') \le 1/(dk)$ for all $r' \in  [r, r + 64(c+4)k \log n)]$. Therefore, for $r' = 64(c+4)k \log n)$, it is true that $\E [p(i,r')] < 1/n^{c+3}$ and, by a Markov bound, nest $n_i$ has at least one ant with probability at most $1/n^{c+3}$. Therefore, with probability at least $1 - 1/n^{c+3}$, $c(i,r') = 0$.
\end{proof}

Next, we show that, for any pair of nests both with population proportions bounded below by $1/(dk)$, we can use Lemma \ref{lem:taylor_series} to argue that the populations of these nests diverge quickly. As soon as a nest drops below the $1/(dk)$ threshold we can use Lemma \ref{lem:small_nest_dropout} to show that it will not be the winning nest.

\begin{lemma}
\label{lem:two-nests-drop}
	Let $n_i$ and $n_j$ be two nests with $q(i) = q(j) = 1$. Then, for $r' = (6 d + 64(c+6) )k \log n$, with probability at least $1 - 1/n^{c+2}$, at least one of the following is true: $c(i, r') = 0$ or $c(j, r') = 0$.
\end{lemma}

\begin{proof}
 Note that if at any point a nest has no ants in it, it remains having no ants thereafter. We consider two possible cases based on how many ants are in each nest in each round $r \in [1, 6 d k \log n)]$:
\begin{itemize}
 \item[Case 1: ] In some round $r \in [1, 6 d k \log n)]$ either $n_i$ of $n_j$ has fewer than $n/(dk)$ ants. Then, by Lemma \ref{lem:small_nest_dropout}, with probability at least $1 - 1/n^{c+3}$, this nest has no ants committed to it after $64(c+4)k \log n)$ rounds. 

\item[Case 2: ] In every round $r \in [1, 6 d k \log n)]$ both $n_i$ and $n_j$ have at least $n/(dk)$ ants. We are going to show that, with probability at least $1 - 1/n^{c+3}$ this case does not happen. Suppose, in contradiction, that in every round $r \in [1, 6 d k \log n)]$ both $n_i$ and $n_j$ have at least $n/(dk)$ ants. First, we want to show that $ \Delta p_L < p_L $, so that we can apply Lemma \ref{lem:taylor_series}. Since both nests have at least $n/(dk)$ ants, it follows that $p_L \geq 1/(dk)$, so the number of ants recruiting for that nest is at least $n/(d^2k^2) = \Omega(\log n)$, by our assumption that $k \leq \sqrt{n/(8 d^2 (c+6) \log n)}  = \BO(\sqrt{n/\log n})$. By a Chernoff bound, with probability at least $1 - 1/n^{c+3}$, the number of recruiting ants is at most $2n/(d^2k^2)$. This indicates that, with probability at least $1 - 1/n^{c+3}$, not all ants from the nest with lower population recruit, and so $ \Delta p_L < p_L $.

This means that by Lemma \ref{lem:taylor_series}, 

\begin{equation*}
\E [\epsilon(i,j, dk \log n)] \geq \left (1+\frac{1}{2dk} \right)^{6 d k \log n} \E [\epsilon(i,j, 1)] \geq n^{3} \cdot \E [\epsilon(i,j, 1)].
\end{equation*}

By Lemma \ref{lem:eps-bound}, $\E [\epsilon(i,j, 1)] \geq 1/(3(n-1))$. Therefore, after $6 d k \log n$ rounds we have $ \E [\epsilon(i,j, 6 d k\log n)] \geq n$. However, this is a contradiction to the fact that $\epsilon(i,j, 6 d k\log n) \leq n-1$, by definition, and consequently $\E[\epsilon(i,j, 6 d k\log n)] \leq n-1$.
\end{itemize}
Each case holds with probability at least $1 - 1/n^{c+3}$, so union bounding them, we get that, with probability at least $1 - 1/n^{c+2}$, either $c(i, r') = 0$ or $c(j, r') = 0$.
\end{proof}

\begin{theorem}\label{thm:simple_runtime}
	With probability at least $1 - 1/n^c$, Algorithm \ref{algo:simple} solves the \textsc{HouseHunting} problem in $\BO(k\log n)$ rounds.
\end{theorem}
\begin{proof}
By our bound on $k$ ($k \leq \sqrt{n/(8 d^2 (c+6) \log n)}  = \BO(\sqrt{n/\log n})$), with probability at least $1 - 1/n^{c+1}$, in the first round of the algorithm (a round of searching), at least some ant will arrive at a nest with quality $1$. Further, since only ants at nests with quality $1$ choose to recruit, at any point, at least one ant is committed to a good nest.

By Lemma \ref{lem:two-nests-drop}, with probability at least $1 - 1/n^{c+2}$, for each pair of nests $n_i$ and $n_j$, at least one nest contains no ants by the end of $\BO(k\log n)$ rounds. Union bounding over all, at most $\binom{k}{2} < k^2 < n$ (by the bound on $k$), pairs of good nests, we conclude that, with probability at least $1 - 1/n^{c+1}$, after $\BO(k \log n)$ rounds, each pair contains at least one nest with no ants. This can only be true if all nests have no ants (not possible) or if all ants are located at one good nest. 

Finally, union bounding the initial search phase and the subsequent competition between nests, we get that with probability at least $1 - 1/n^c$, the house hunting problem is solved in $\BO(k\log n)$ rounds.
\end{proof}

\section{Discussion and Future Work}
\label{sec:discussion}
	
	\paragraph{Extensions to the Model}
	
	For the sake of analysis, we have made many simplifying assumptions about the house-hunting process. We are confident that many of these assumptions can be weakened to make the model more realistic and natural. Some obvious modifications include assuming ants know only an approximation of $n$, allowing values of $k$ larger than $\BO(n/\log n)$, and allowing non-binary nest qualities, variability in the ants' quality sensing, along with some measure of algorithmic performance based on the quality of the chosen nest. Distinguishing between direct transport and tandem runs may also be interesting, paired with a more fine-grained runtime analysis. 
	
	Additionally, real ants can only assess nest quality and population approximately. For example, they may estimate nest size (one measure of quality)  by randomly walking within the nest and counting how many times they cross their previous path \cite{mallon2000ants}. They seem to estimate nest population by measuring encounter rates with other ants, with a higher encounter rate indicated a higher population at the nest \cite{gordon2010ant,pratt2010nest}. Adding noisy measurements to our model and designing algorithms that handle this noise would be a very interesting future direction. It may even be possible to explicitly model lower level behavior and implement subroutines for nest assessment, population measurement, recruitment, and search which give various runtime and error guarantees.
	
	\paragraph{Extensions to the Algorithms}
	
	We believe that Algorithm \ref{algo:simple} may be a good starting point for work on more realistic house-hunting models. Below we discuss some interesting possible extensions to the algorithm. Some seem to simply require a more involved analysis, while others seem to require trade-offs in the algorithm's running time and its level of simplicity.
	
	\textbf{Improved running time: } The $\BO(k \log n)$ runtime of Algorithm \ref{algo:simple} is required because, on average each nest initially contains $n/k$ ants, so ants only recruit with probability $1/k$. $\BO(k)$ time is required to amplify population gaps by a constant factor. Ideally, ants would all recruit with a probability lower bounded by a constant, but still linearly dependent on the nest populations. This would allow convergence in $\BO(\log n)$ rounds. If ants keep track of the round number, they can map this to an estimate $\tilde k(r)$ of how many competing nests remain, allowing them to recruit at rate $\BO(c(i,r)/n \cdot \tilde k(r))$. We believe that such a strategy should yield a relatively natural algorithm converging in $\BO(\log^c n)$ rounds.
	
	\textbf{Non-binary nest qualities: } Assuming a real-valued nest quality in the range $(0,1)$ affects the correctness  of Algorithm \ref{algo:simple} because ants no longer have the notion of a good nest. However, it should be possible to incorporate the quality of the nest into the recruitment probability in order make the algorithm converge to a high-quality nest, without significantly effecting runtime.
	
	\textbf{Approximate counting, nest assessment, and knowledge of $\bv{n}$: } Since the analysis of Algorithm  \ref{algo:simple} does not require each ant to recruit with a specific probability, but rather that the total number of ants choosing to recruit from a nest is proportional to its population, it should be resilient to noisy quality and population measurements. As long as ants have unbiased estimators of these values, we believe that Algorithm \ref{algo:simple} should remain correct, perhaps with some runtime cost dependent on estimator variance.
	
	\textbf{Fault tolerance: } Similarly, Algorithm \ref{algo:simple} should support some degree of fault tolerance. A small number of ants suffering from crash-faults or even malicious faults, should not affect the overall populations of recruiting ants and the algorithm's performance.
		
	\textbf{Asynchrony: } Finally, note that Algorithm \ref{algo:simple} currently works in synchronous rounds and relies on that assumption to get the correct number of ants at a given nest. However, we believe that, as long as the distribution of ants in candidate nests throughout time stays close to the distribution in the synchronous model, Algorithm \ref{algo:simple} can be extended to work in a partially-synchronous model, potentially at the cost of some extra running time.

\bibliographystyle{plain}
\bibliography{references}

\appendix

\section{Math Preliminaries}

\begin{lemma}
\label{lem:dependent-chernoff}
	Let $X_1, \cdots, X_n$ be arbitrary binary random variables. Let $X^*_1, \cdots, X^*_n$ be random variables that are mutually independent and such that for all $i$, $X^*_i$ is independent of $X_1, \cdots, X_{i-1}$. Assume that for all $i$ and all $x_1, \cdots, x_{i-1} \in \{0,1\}$,
	
\begin{equation*}	
	 P[X_i = 1 | X_1 = x_1, \cdots, X_{i-1} = x_{i-1}] \geq P[X^*_i = 1]
\end{equation*}

Then, for all $k \geq 0$, we have,

\begin{equation*}
	P \left[ \sum_{i=1}^n X_i < k \right] \leq P \left[ \sum_{i=1}^n X^*_i < k \right],
\end{equation*}

and the latter term can be bounded by Chernoff bounds for independent random variables.
\end{lemma}

\end{document}